\documentclass{article}

\usepackage{arxiv}

\usepackage[utf8]{inputenc} 
\usepackage[T1]{fontenc}    
\usepackage{hyperref}       
\usepackage{url}            
\usepackage{booktabs}       
\usepackage{amsfonts}       
\usepackage{nicefrac}       
\usepackage{microtype}      
\usepackage{lipsum}
\usepackage{amsmath, amssymb}
\usepackage{comment}
\usepackage{multirow}
\usepackage{graphicx}
\usepackage{comment}
\usepackage{tikz}
\usepackage{pgfplots}
\usepackage{amsfonts}

\usetikzlibrary{math} 
\pgfplotsset{compat=newest}

\newtheorem{theorem}{Theorem}
\newtheorem{lemma}[theorem]{Lemma}
\newtheorem{proof}{Proof}
\newtheorem{definition}{Definition}
\newtheorem{problem}{Problem}

\newcommand{\Corr}{\mathrm{Corr}}
\newcommand{\sign}{\mathrm{sign}}

\title{NAPLES;Mining the lead-lag Relationship from Non-synchronous and High-frequency Data}

\author{
  Katsuya Ito\\
  Preferred Networks, Inc.\\
  Otemachi Bldg., 1-6-1 Otemachi, Chiyoda-ku, Tokyo, 100-0004, Japan \\
  \texttt{katsuya1ito@preferred.jp}
   \And
  Kei Nakagawa\\
  Innovation Lab\\
  Nomura Asset Management Co., Ltd.\\
  1-11-1 Nihonbashi, Chuo-ku, Tokyo, 103-8260, Japan \\
  \texttt{kei.nak.0315@gmail.com} 
}

\begin{document}
\maketitle

\begin{abstract}
In time-series analysis, the term "lead-lag effect" is used to describe a delayed effect on a given time series caused by another time series. 
lead-lag effects are ubiquitous in practice and are specifically critical in formulating investment strategies in high-frequency trading. 
At present, there are three major challenges in analyzing the lead-lag effects.
First, in practical applications, not all time series are observed synchronously.
Second, the size of the relevant dataset and rate of change of the environment is increasingly faster, and it is becoming more difficult to complete the computation within a particular time limit.
Third, some lead-lag effects are time-varying and  only last for a short period, and their delay lengths are often affected by external factors. 
In this paper, we propose NAPLES (Negative And Positive lead-lag EStimator), a new statistical measure that resolves all these problems. 
Through experiments on artificial and real datasets, we demonstrate that NAPLES has a strong correlation with the actual lead-lag effects, including those triggered by significant macroeconomic announcements.
\end{abstract}


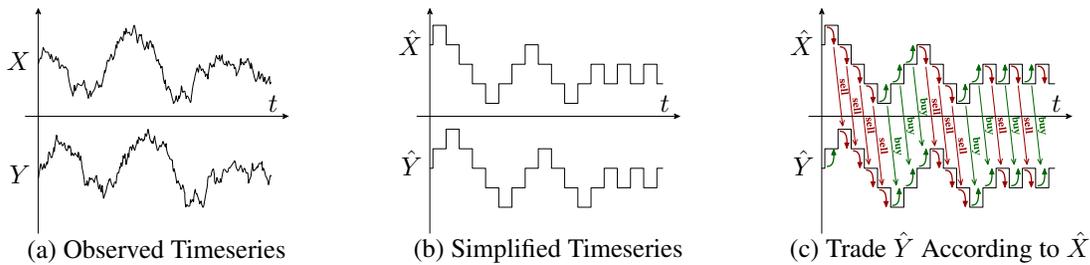
\begin{figure}[ht]\label{fig_idea}
\centering
 \begin{minipage}{0.31\hsize}
\begin{tikzpicture}[scale=1]
\pgfdeclareimage[width=3.5cm,height=3cm]{fig-wave-a}{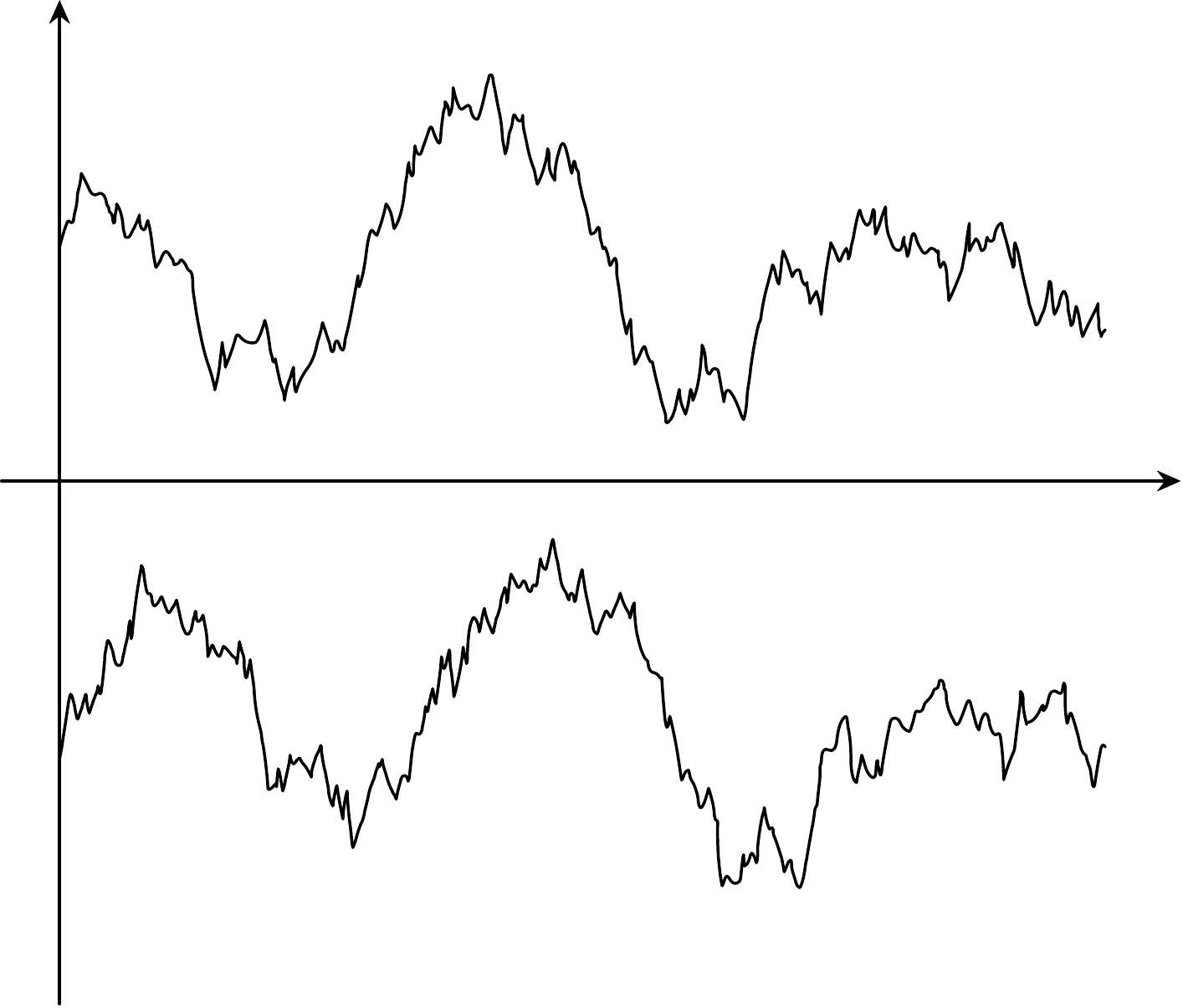}
\node (IMAGE){\pgfuseimage{fig-wave-a}};
\draw (-1.55cm, 0.8cm) node [left] {$X$};
\draw (-1.55cm, -0.7cm) node [left] {$Y$};
\draw (1.55cm, 0cm) node [above] {$t$};
\draw (0cm, -2.0cm) node [above] {(a) Observed Timeseries};

\end{tikzpicture}
\end{minipage}
\begin{minipage}{0.31\hsize}
\begin{tikzpicture}[scale=1]
\pgfdeclareimage[width=3.5cm,height=3cm]{fig-wave-b}{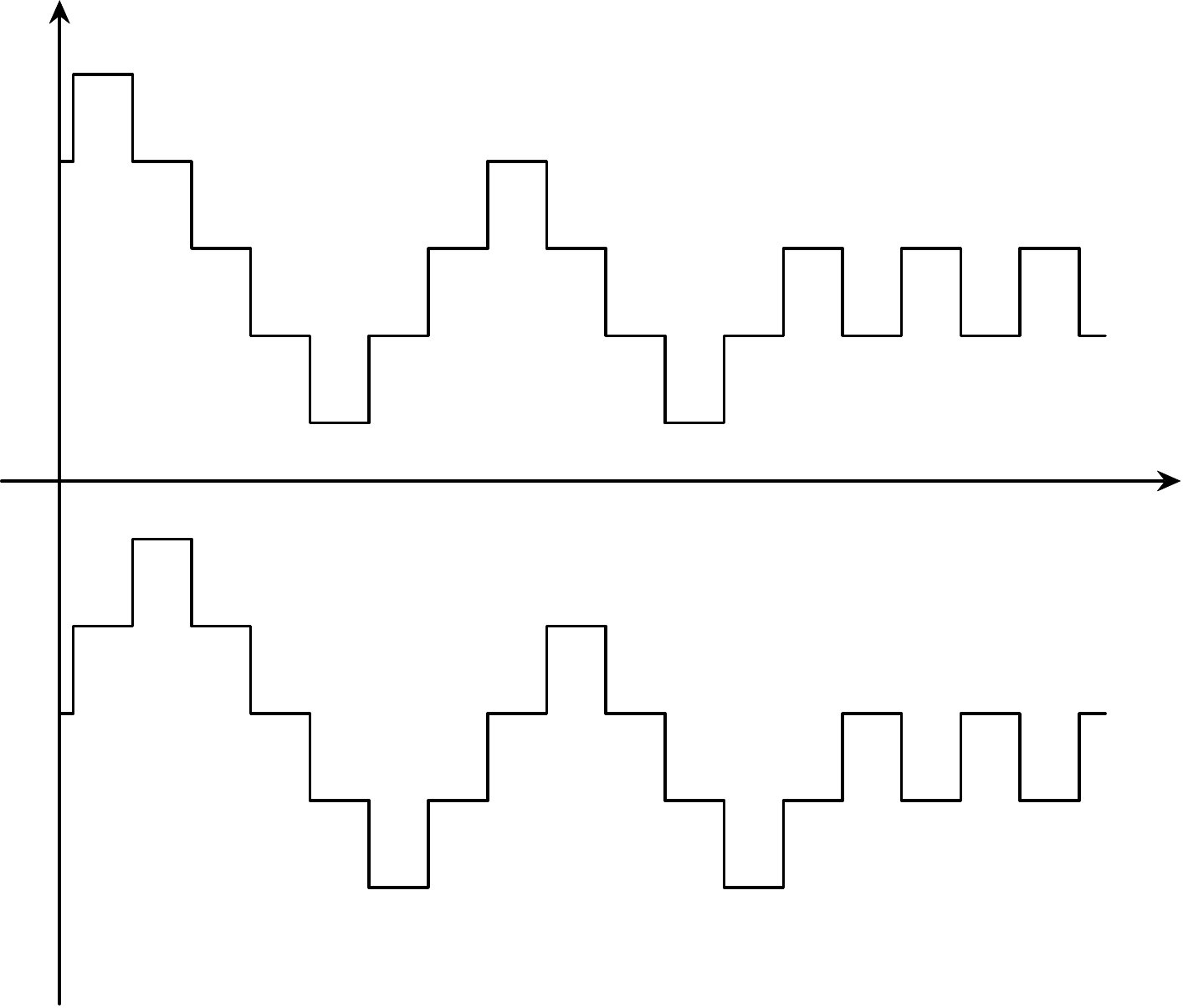}
\node (IMAGE){\pgfuseimage{fig-wave-b}};
\draw (-1.55cm, 1.05cm) node [left] {$\hat{X}$};
\draw (-1.55cm, -0.55cm) node [left] {$\hat{Y}$};
\draw (1.55cm, 0cm) node [above] {$t$};
\draw (0cm, -2.0cm) node [above] {(b) Simplified Timeseries};
\end{tikzpicture}
\end{minipage}
 \begin{minipage}{0.31\hsize}

\begin{tikzpicture}[scale=1]
\pgfdeclareimage[width=3.5cm,height=3cm]{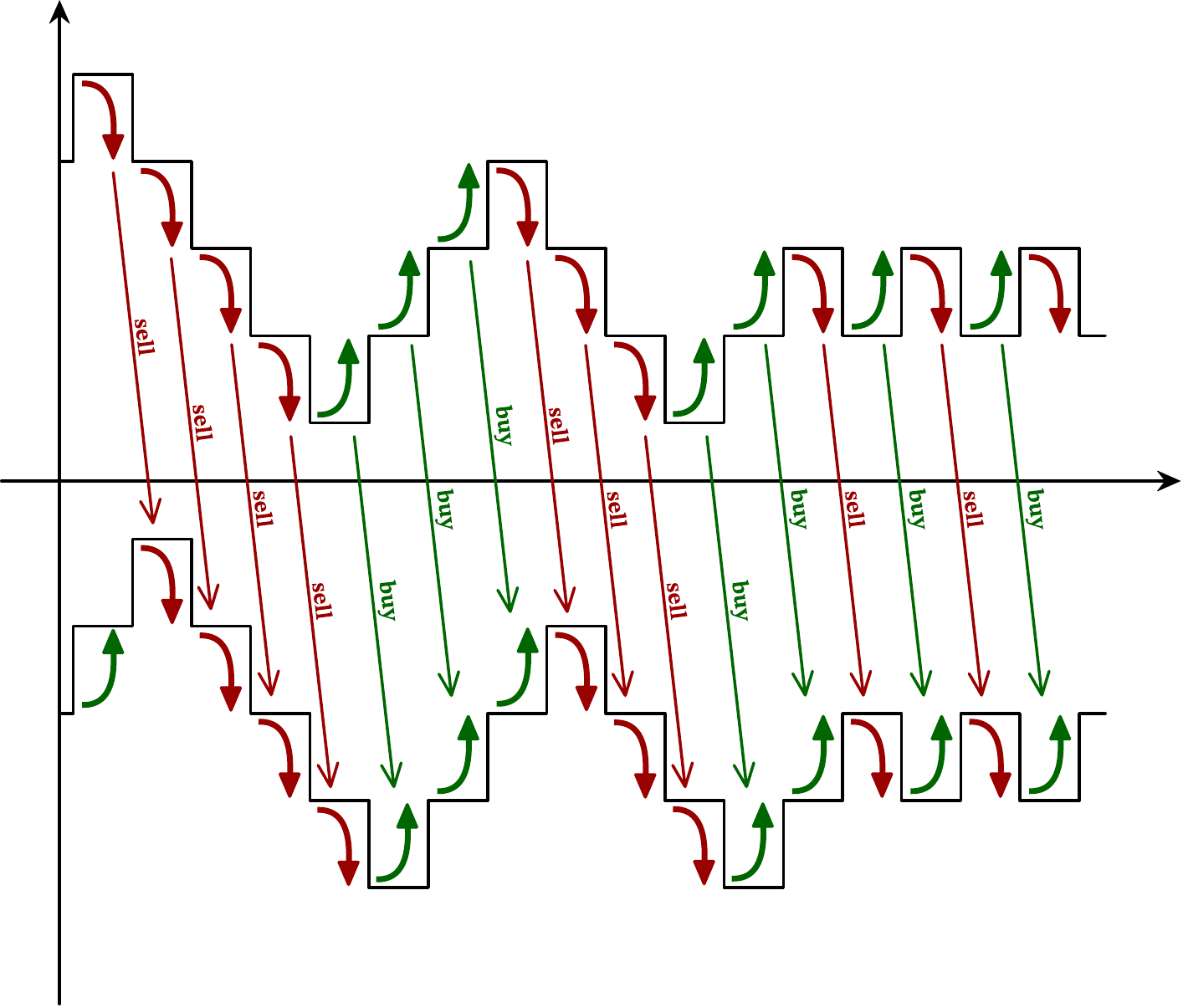}{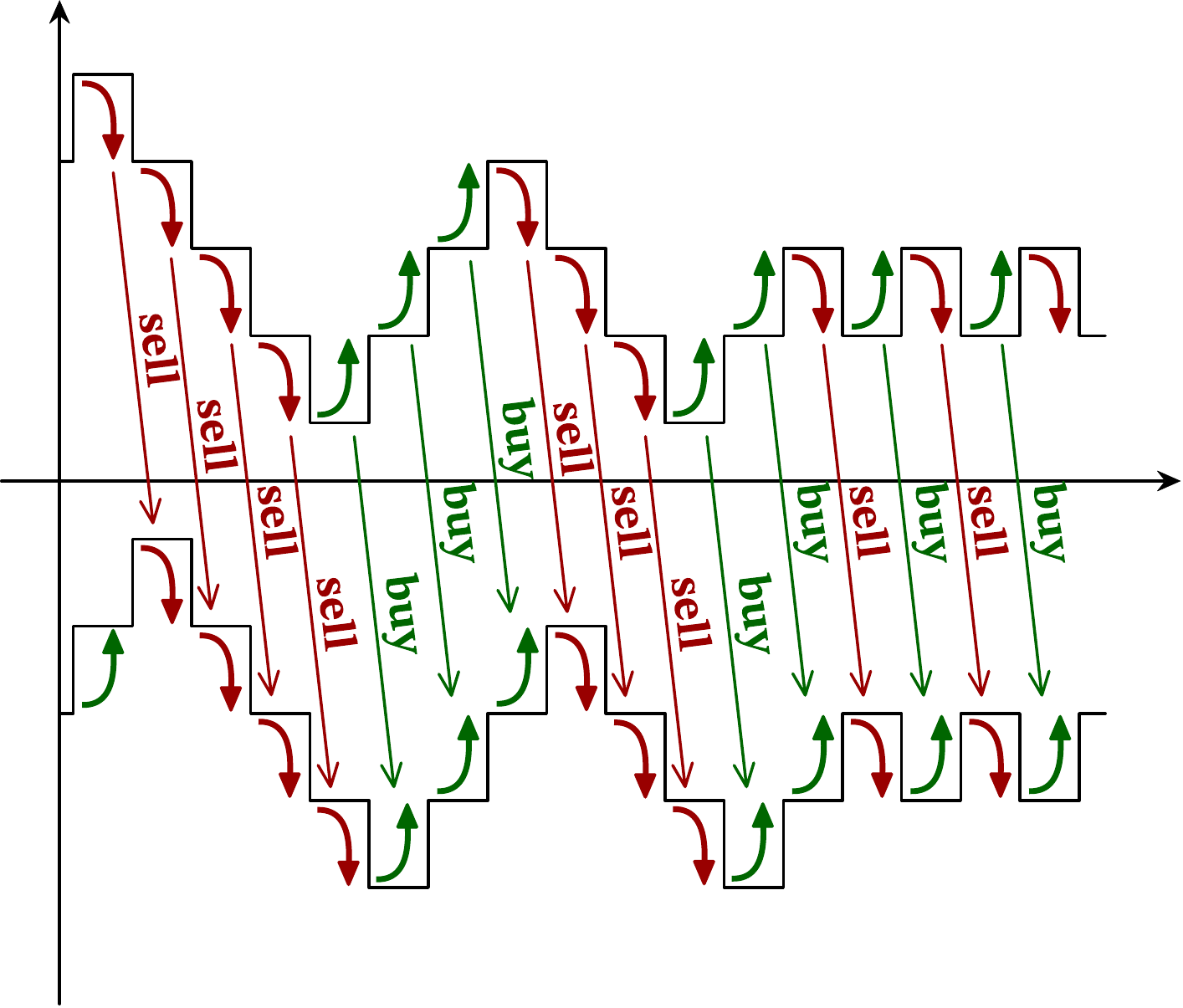}
\node (IMAGE){\pgfuseimage{fig-wave-c}};
\draw (-1.55cm, 1.05cm) node [left] {$\hat{X}$};
\draw (-1.55cm, -0.55cm) node [left] {$\hat{Y}$};
\draw (1.55cm, 0cm) node [above] {$t$};
\draw (0cm, -2.0cm) node [above] {(c) Trade $\hat{Y}$ According to $\hat{X}$};

\end{tikzpicture}
\end{minipage}
\caption{Illustration of our method.
(a $\to$ b) We simplify the observed time series using a sign.
(b $\to$ c) We estimated the lead-lag by making a trading strategy using the simplified time series.
In the trading strategy, we trade $\hat{Y}$ when $\hat{X}$ moves, and we can make profit if $X$ is leading $Y$.}

\end{figure}

\section{Introduction}
The lead-lag effect is a phenomenon wherein the change of one time series affects another time series with a delay. 
This effect appears in various fields from financial markets to control engineering.
In particular, the lead-lag effect in financial markets has been thoroughly investigated 
because it is a crucial factor in investment, and the knowledge of the size of the lag can potentially help investors make a profit without taking any risk.
\par
Discovering lead-lag effects is vital because high-frequency trading uses high-speed communication, and the high-performance computers used in high-frequency trading can make profits by using unapparent lead-lag effects in millisecond and nanosecond timescales \cite{Biais2015,Wang2015,Martinez2012}.
Moreover, high-frequency trading also makes the market more "efficient."
In high-frequency trading, the lead-lag effect does not usually persist for a long time because market investors identify such effects, and are often quick to take advantage of such asymmetry in the market
\cite{CHABOUD2014,Case2011,Bhaskkar2008}.
Therefore, to make a profit, financial analysts must be quick and accurate in estimating the lead-lag effect.

To detect the lead-lag effect, we formulate the following maximization problem(eq.\ref{main}).
Let $X_t,Y_t$ be asset prices, and we assume that $X_t$ is observed at $s_1 < s_2 < \cdots < s_n$ and $Y_t$ is observed at $t_1 < t_2 < \cdots < t_m$.
We say that $X_t$ leads $Y_t$ with \emph{a lead-lag parameter} $\theta$ if the correlation between $X_t$ and $Y_{t+\theta}$ is stronger than the correlation between $X_t$ and $Y_{t+\theta'}$ for any $\theta' \neq \theta$.
We find the maximum lead-lag parameter $\hat{\theta}$ as follows:
\begin{equation}
\hat{\theta} = \mathrm{argmax}_{\theta \in \Theta} \hat{\Corr}(\{X_{s_i}\}_{i=1}^n,\{Y_{t_j+\theta}\}_{j=1}^m)
\label{main}
\end{equation}

where $\hat{\Corr}(\{X_{s_i}\}_{i=1}^n,\{Y_{t_j}\}_{j=1}^m)$ is a correlation estimator for the observed time series, and $\Theta$ is a candidate of the lead-lag parameters.\par

In estimating the lead-lag effect in high-frequency trading, there are three major challenges: (1) Non-synchronicity of timestamps (2) Computational complexity (3) Time-variance of the lead-lag effect.\par
(1) By the non-synchronicity of timestamps, we mean that the price of each asset is observed independently and irregularly.
In high-frequency trading, price movement occurs irregularly since transactions happen asynchronously and independently of each asset.
Non-synchronicity makes it difficult for us to estimate the correlation. Consequently, it does not allow us to estimate a lead-lag effect by simple methods.
Hoffman, Rosenbaum, \& Yoshida(2013) proposed a solution to this problem using an estimator that has mathematical convergence \cite{Hoffmann2013}.
Moreover, various theoretical analyses
\cite{Chiba2019,Hayashi2017,Hayashi2018}, and empirical analyses \cite{Alsayed2014,Bollen2016,Ceron2016,HuthAbergel2011} have practically applied the method.\par
(2) Computational complexity is a significant challenge in high-frequency trading, for which relevant information is becoming increasingly available.
The number of assets and lag candidates also increase, along with the complexity of the relationships between them.
If the length of the time series is $T$, the number of assets is $N$, and the number of lag candidates is $L$, the current extent of the calculation is in the order of $O(T^2 N^2 L)$. 
A calculation algorithm better than this order has not been proposed.\par

(3) The time-variance of the lead-lag parameter is another critical factor that has not been explored in depth.
The time-variance of the lead-lag relationships means that some lead-lag effects occur temporally. 
Most existing lead-lag estimation methods assume that the lead-lag parameter $\theta$ is constant; therefore, they can be applied only to limited situations.
This time-varying nature of the lead-lag effect has been observed empirically \cite{Huth2014}.\par
In this paper, we propose a new estimator of lead-lag effects to solve these three problems.
This method can not only be performed using a non-synchronously observed set of time series, but it can also be performed with a computational cost that is linear to the number of observations $O(N)$.
The main contributions of this paper are as follows.
\begin{itemize}
    \item We propose an index that estimates the lead-lag effects in high-frequency trading, and theoretically prove that our index can estimate lead-lag.
    \item We demonstrate the effectiveness of our method by using high-frequency simulated data and market data.
Specifically, we show that the estimation error is smaller, and the convergence speed is faster than in other methods.
    \item We use our method to detect the existence of the temporal lead-lag effect when there are significant macroeconomic announcements with an experiment using foreign exchange market data.
\end{itemize}
\section{Problem Formulation and Related Work}
\subsection{Notation}
We denote the interval of the numbers between $a \in \mathbb{R}$ and $b \in \mathbb{R}$ as follows:
$[a,b]=\{x | a\le x \le b\}$ ,$]a,b] = \{x| a < x \le b\}$, $[a,b[\  = \{x | a \le x < b\}$, $]a,b[ \ = \{x| a<x<b\}$.
$a\vee b$ and $a \wedge b$ stand for operations $\max\{a,b\}$ and $\min\{a,b\}$, respectively. 
$1_A$ denotes an indicator function which is $1$ when the condition $A$ is true and $0$ when $A$ is false.
$\theta$ denotes the lead-lag delay parameter between two time series, and $\rho$ denotes the correlation between the two time series.


\subsection{Problem Formulation}
We formulate the lead-lag estimation problem mathematically.
First, we define the notion of the lead-lag effects and non-synchronous observations.
Let $X_t,Y_t$ be two stochastic processes 
(e.g., asset returns).
Assume that $X_t$ is observed at $0=s_1<s_2<\cdots<s_n=T$ and $Y_t$ is observed at $0=t_1<t_2<\cdots<t_m=T$.

\begin{definition}
$X_t$ is said to be \textbf{leading} $Y_t$ with the \textbf{lead-lag parameter} $\theta$
If for all $\alpha$, the absolute value of the correlation between $X_t,Y_{t+\alpha}$ is larger than that of $X_t,Y_{t+\theta}$
\end{definition}

\begin{definition}
$X_t$ and $Y_t$ are said to be observed \textbf{synchronously} if $s_i = t_i$ holds for all $i$ and are said to be observed \textbf{non-synchronously} if they are not observed synchronously.
\end{definition}
Our main problem is to infer the lead-lag effect from non-synchronously observed data.
\begin{problem}
Assume that $X_t$ leads $Y_t$ with the lead-lag parameter $\theta$.
Our problem is to infer the lead-lag parameter $\hat{\theta}$ only from non-synchronously observed data $\{X_{s_i}\}_{i=0}^n$ and $\{Y_{t_j}\}_{j=0}^m$
\end{problem}

\subsection{Related Work}
Many studies, such as \cite{KAWALLER1987} and \cite{Lo1990}, showed the existence of the lead-lag effect for synchronously observed time series.
Because high-frequency trading has become the standard trading method, the importance of  modeling and analysis of non-synchronously observed data is increasing both in theory and in practice.
Hayashi and Yoshida \cite{Hayashi2005} proposed the following covariance estimator.

\begin{definition}[\cite{Hayashi2005}]
Assume that $X_t$ is observed at $0=s_1<s_2<\cdots<s_n=T$ 
and $Y_t$ is non-synchronously observed at $0=t_1<t_2<\cdots<t_m=T$.
Then, the Hayashi-Yoshida covariance estimator is defined as follows.
\[
\mathrm{HY}(X_s,Y_t,\{s_i\}_{i=1}^n,\{t_j\}_{j=1}^m):=
\]
\[
\sum_{i=1}^{n-1} \sum_{j=i}^{m-1} 
(X_{s_{i+1}}-X_{s_{i}})(Y_{t_{j+1}}-Y_{t_{j}})
1\{{]t_{i-1},t_{i}] \cap ] s_{ j-1 } , s_{ j } ]\neq\emptyset}\}
\]
\end{definition}
They proved the convergence of this estimator when $X_t,Y_t$ are in Brownian motion \cite{Hayashi2005}.

On the other hand, Hoffman et al. \cite{Hoffmann2013} showed that the covariance estimator with a transferred time series $U^n(\tilde{\theta})$ can estimate the lead-lag effect when $X_t,Y_t$ are in Brownian motion.

\[
U^n(\tilde{\theta})
: = 1 _ { \tilde { \theta } \geq 0 }
\mathrm{HY}(X_{s-\tilde{\theta}},Y_t,\{s_i-\tilde{\theta}\}_{i=1}^n,\{t_j\}_{j=1}^m)
+
\]
\[
1 _ { \tilde { \theta } < 0 } 
\mathrm{HY}(X_{s+\tilde{\theta}},Y_t,\{s_i+\tilde{\theta}\}_{i=1}^n,\{t_j\}_{j=1}^m).
\]

Chiba~\cite{Chiba2019} extended this estimator to the case when $X_t,Y_t$ are partially in Brownian motion. 
Hayashi and Koike \cite{Hayashi2017} improved the estimator with a Fourier transform and Hayashi and Koike \cite{Hayashi2018} improved it with a Wavelet transform.\par

In contrast, Dobreva and Schaumburgb \cite{DobrevaSchaumburgb2017} introduced the following lead-lag estimation methods.

\begin{definition}[\cite{DobrevaSchaumburgb2017}]
Assume that two asset prices $X$,$Y$ are observed.
$
1\{ \mbox{Z active\ at\ t} \}
$
is a function that returns $ 1 $ when a certain asset $ Z \in \{X, Y\} $ is traded at time $ t $, otherwise $ 0 $.
Here, the Dobreva-Schaumburgb index $ DS_t $ is calculated as follows.
\[
DS_t:=
\frac{
\sum _ { i = | t | } ^ { N - | t | } 
1 \{ \mbox{A active\ at\ } i+t \}
\cap 
1\{ \mbox {B  active at\ } i \}
}{
\sum _ { i = | t | } ^ { N - | t | } 
1 \{ \mbox{A active\ at\ } i+t \}
\wedge
\sum _ { i = | t | } ^ { N - | t | } 
1 \{ \mbox{B active\ at\  } i \}
}.
\]
\end{definition}

We can estimate the lead-lag effect with the Dobreva-Schaumburgb index in the same way as \cite{Hoffmann2013}. 
Its statistical properties are not apparent, but empirical studies show its effectiveness.

\section{Methodology}
\subsection{NAPLES : Our Estimator of lead-lag Effect}
We define the index $R(t)$ for the lead-lag effect estimation, which is the main method of this paper.

Before defining our lead-lag estimator, we prepare some notations.
Assume that we observe $\{Z_{u_k}\}_{k=1}^n \in \{\{X_{s_i}\}_i,\{Y_{t_j}\}_j\}$.
Let $r^{(Z)}_{u_k}$ as the logarithmic returns of assets $Z$,
$
r^{(Z)}_{u_k}=\log (Z_{u_{k}}/Z_{u_{k-1}}),\
$.
Let $b^{(Z)}_{u_k}$ be the sign of $r^{(Z)}_{u_k}$, 
$
b^{(Z)}_{u_k}=\sign \{r^{(Z)}_{u_k}\}$.
Then we denote $\hat{Z_t}$ as the cumulative sum of the signs of returns $b^{(Z)}_{u_k}$,$
\hat{Z_t} = \sum_{k=1}^n 1_{u_k < t} b_{u_k}^{(Z)},\ \ 
$
\begin{definition}
(NAPLES; Negative And Positive lead-lag EStimator)
Assume that $X_t$ is observed at $0=s_1<s_2<\cdots<s_n=T$ 
and $Y_t$ is non-synchronously observed at $0=t_1<t_2<\cdots<t_m=T$. 
we define our index for lead-lag estimation as follows:

\[
R(t;X_t,Y_t):=
\sum_{i=1}^{n-1} 
\left(
b^{(X)}_{s_{i}}\hat{Y}_{s_{i+1}} -b^{(X)}_{s_{i}}\hat{Y}_{s_{i}} 
\right)
1_{s_{i+1}<t}
-
\sum_{j=1}^{m-1} 
\left(
b^{(Y)}_{t_{j}}\hat{X}_{t_{j+1}} - b^{(Y)}_{t_{j}} \hat{X}_{t_{j}} 
\right)
1_{t_{j+1}<t}
\]
\end{definition}

This definition looks complicated but is based on a simple idea.
If $X$ is leading $Y$, trading $Y$ following the price movement of $X$ makes a profit or a loss.
$R(t)$ is an index that expresses this idea numerically.
The first term
\[
b^{(X)}_{s_ {i}} \hat{Y}_{s_{i + 1}} -b ^ {(X)}_{s_ {i}} \hat{Y}_{ s_ {i}}
\]
is the  return of the trading strategy which buys $\hat{Y}$ at $s_i$ and sells $\hat{Y}$ at $s_{i+1}$ when $X$’s price rises at time $s_i$ and does the opposite when the price of $X$ falls.

The second term
\[
b^{(Y)}_{t_{j}} \hat{X}_{t_{j + 1}}-b^{(Y)}_{t_{j}} \hat{X}_{t_{j}}
\]
is the return of the trading strategy which buys $\hat{X}$ at $t_j$ and sells $\hat{X}$ at $t_{j+1}$ when $Y$’s price rises at time $t_j$ and does the opposite when the price of $Y$ falls.
Figure \ref{fig_idea} is the illustration of our idea.

This method solves the three problems described above.
(1) For the issue of non-synchronicity, we can use this method for time series of non-synchronous observations. 
(2) To solve computational complexity, our estimator is simply a sum of $N$ terms; therefore, we can calculate this estimator with the order $O(N)$. 
(3) To address the time variance of the lead-lag effect, our estimator is a return of a trading strategy so that we can calculate the lead-lag parameter at each time.

To clarify the mathematical properties of $R(t)$, we present the following theorem.

\begin{theorem}\label{mainthm}
Assume that two asset prices $X, Y$ follow geometric Brownian motion with the correlation $\rho$ and have a lead-lag relationship of $\theta$. 
In addition, we assume that the observation times $s_i$ and $t_j$ are always equally spaced $\Delta $. 
Then, $ R (T) $ is as follows. \\
\[
\mathrm{E}[R(T;X_t,Y_t)]=\left\{\begin{array}{ll}{
\frac{l}{\pi} \sign(\theta) \arcsin \left(\frac{\rho|\theta|}{\Delta}\right)} & {\text{if } 0<|\theta|<\Delta } \\ 
{\frac{l}{\pi} \sign(\theta) \arcsin \left(\frac{\rho(2 \Delta-|\theta|)}{\Delta}\right)} & {\text{if } \Delta<|\theta|<2 \Delta } \\
{0} & {\text{otherwise}}\end{array}\right.
\]
\end{theorem}

We show the general case of this theorem and its proof in the Appendix.
Thus, we can estimate the lead-lag parameter $\theta$ by calculating
\begin{equation}
\hat{\theta} = \mathrm{argmax}_{\theta \in \Theta} 
R(T,X_t,Y_{t+\theta})
\label{leadlagest}
\end{equation}

\subsection{Interpretation of The Main Theorem}
We interpret the main Theorem \ref{mainthm}.
In Equation \ref{leadlagest}, we consider a situation where we want to check if the lead-lag effect exists for a given $X_t$ and $Y_{t+\alpha}$.
Figure \ref{graph:ert} shows that our estimator achieves the maximum value at $\Delta$ and shows that we can estimate the lead-lag parameters by using argmax like in Equation \ref{leadlagest}.
We show the meaning of each point in Figure \ref{graph:ert}.
At point C, there is no lead-lag effect between $X$ and $Y$, so there is no profit if one trades $Y$ according to the price movement of $X$.
At point D, there is a lag equal to the trading interval, so we can always make a profit by trading $Y$ $\Delta$ seconds after $X$ moves.
At point B, we perform the opposite trade; thus, a loss occurs.
On the other hand, at point E and A, if we trade two periods after the price movement of X, we cannot make a profit because $X$ and $Y$ are in Brownian motion, and their increments are independent.
In other words, $X_{t+\Delta} - X_t$ and $X_{t+2\Delta} - X_{t+\Delta}$ are independent.
\begin{figure}[ht]
\centering
\begin{tikzpicture}[scale=1]

\begin{axis}[ 
xlabel={$\theta$}, ylabel={$E[R(T)]$}
,axis lines=middle
,samples=41, thick
,domain=-4:4
,ymin=-10, ymax=10,
,yticklabels={,,}
,xticklabels={,,}
,extra x ticks={1,-1,2,-2}
,extra x tick labels={$\Delta$,$-\Delta$,$2\Delta$,$-2\Delta$}
,extra x tick style={xticklabel style={yshift=0.5ex,xshift=0.6ex, anchor=south}
}
]
\pgfmathsetmacro{\l}{0.26}
\pgfmathsetmacro{\p}{3.141592}
\pgfmathsetmacro{\r}{1.0}
\pgfmathsetmacro{\d}{1.0}
\addplot+[no marks,domain=0:\d,black]{((\l)/\p)*asin((\r*x)/\d)};
\addplot+[no marks,domain=-\d:0,black] {-((\l)/\p)*asin((-\r*x)/\d)};
\addplot+[no marks,domain=\d:2*\d,black]{((\l)/\p)*asin((\r*(2*\d-x)/\d))};
\addplot+[no marks,domain=(-2)*\d:(-1)*\d,black]{-((\l)/\p)*asin((\r*(2*\d+x)/\d))};
\addplot+[no marks,domain=2*\d:4*\d,black]{0};
\addplot+[no marks,domain=(-4)*\d:(-2)*\d,black]{0};
\draw (0,0)node[below right]{C}; 
\draw (\d,\l*31)node[below right]{D}; 
\draw (-\d,-\l*30)node[below right]{B}; 
\draw (-2*\d-0.3,0)node[below right]{A}; 
\draw (2*\d,0)node[below right]{E}; 



\end{axis}
\end{tikzpicture}
\caption{$E[R(T)]$}\label{graph:ert}
\end{figure}


\section{Experiment}
In this section, we experiment with simulated data and real market data to show the effectiveness of our proposed method. 
We compare our method with the Hoffman-Rosenbaum-Yoshida method and the Dobreva-Schaumburgb method\cite{Hayashi2018}.

\subsection{Comparison of Estimation Error and Convergence Speed}
We use the following paired geometric Brownian motion in the experiment.
\[
\left\{
\begin{array}{l}{X_{t}=x_{0}\exp\left(\sigma_{1}B_{t}^{(1)}\right)}\\
{Y_{t}=y_{0}\exp\left(\rho\sigma_{2}B_{t-\theta}^{(1)}+\sigma_{2}\left(1-\rho^{2}\right)^{1/2}W_{t-\theta}\right)}
\end{array}
\right.
\]

We apply the following parameters: the initial value of each asset $ x_0 = y_0 = 100$,
volatility of each asset $ \sigma_1 = \sigma_2 = 1.0, \ $,
lead-lag parameter $ \theta = 10 $, correlation $ \rho = 0.9 $,
candidates of lead-lag parameter $\Theta= \{-100, -99, \cdots, 99,100 \}$.
We observed $X_t$ and $Y_t$ at $\{s_i\}$ and $\{t_j\}$ and $s_i - s_{i-1} \sim N(10,2)$, $t_j - t_{j-1} \sim N(10,2)$.
We investigated how long each algorithm takes to converge by changing $T$, which is the maximum value of observation time and we use $T \in \{10^{2.5}, 10^{2.6}, \cdots, 10^{5.0} \} $. 
We evaluated the Mean Absolute Error of the estimation and performed the experiments 1000 times generating random numbers, $n=1000$ times. Let $ \hat{\theta}_i $ be  the lead-lag we estimated by each algorithm in the $i$-th experiment and let the real lead-lag be $\theta_i $, then we calculated the average of the absolute errors $ \frac{1}{n} \sum_ {i = 1}^n | \hat{\theta}_i-\theta_i |$.
Figure \ref{fig_convergence} shows the log of the Mean Absolute Error for each $\log_{10} T$ and that our method converges best to the real value $\theta$.

\begin{figure}
\centering
\includegraphics[width=0.8\textwidth]{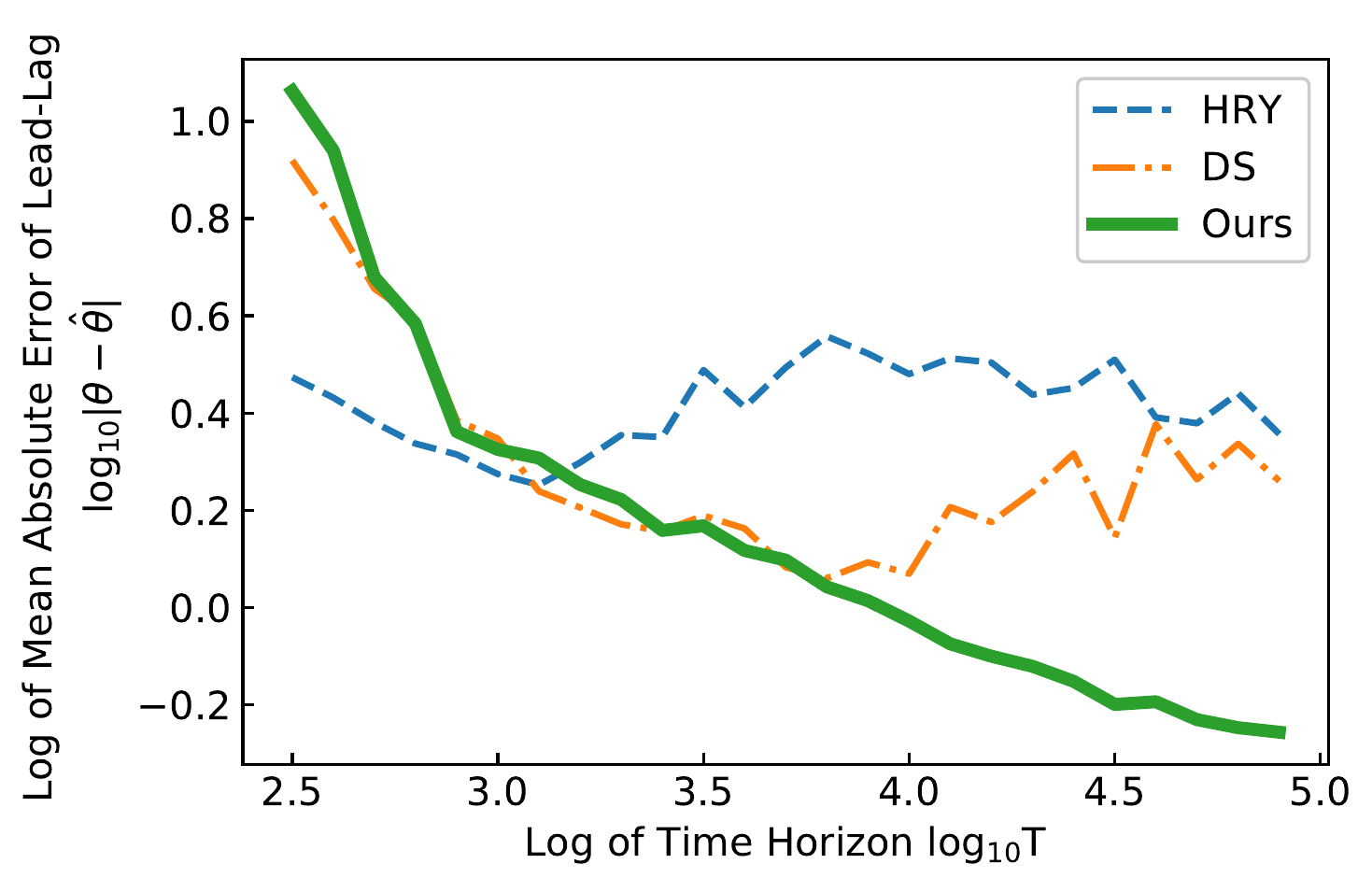}
\caption{
Each line shows the log of the Mean Absolute Error of each estimator.
The figure shows that our method converges best to the real value.
} \label{fig_convergence}
\end{figure}

\subsection{lead-lag estimation of non-synchronously observed foreign exchange data}
In this section, we demonstrate that our method can estimate lead-lag effects with actual market data.
We estimated lead-lag parameters with actual foreign exchange market data.
We chose two currency pairs AUD / USD (Australian Dollar / United States Dollar) and NZD / USD (New Zealand Dollar / United States Dollar), which are geographically close and are categorized as similar assets for international investors. 
We used the Tick-by-Tick data from January 1, 2019, to September 30, 2019, which were downloaded from Dukascopy Historical Data Feed \footnote{\url{https://www.dukascopy.com/}}.\par
We used "Economic calendar Investing.com Forex (2011-2019)" for the list of important announcements in Australia and New Zealand \footnote{https://www.kaggle.com/devorvant/economic-calendar}.
Figure \ref{fig:nzd} and Figure \ref{fig:aud} illustrate the lead-lag parameter before and after the important announcements.
Moreover, the insights from these figures are as follows:

\begin{enumerate}
    \item lead-lag relationship changes before the important announcements.
    \item Our index can find the lead-lag earlier than Hayashi-Yoshida's estimator. (a1, a2 and b4)
    \item The lead-lag estimated by our index is almost similar in terms of the sign as that of the Hayashi-Yoshida estimator.
    \item Our indicator can find lead-lag effects that the Hayashi-Yoshida estimator cannot find (a3 and b2 and b4).
\end{enumerate}

\begin{figure}
\centering
   \includegraphics[width=\textwidth]{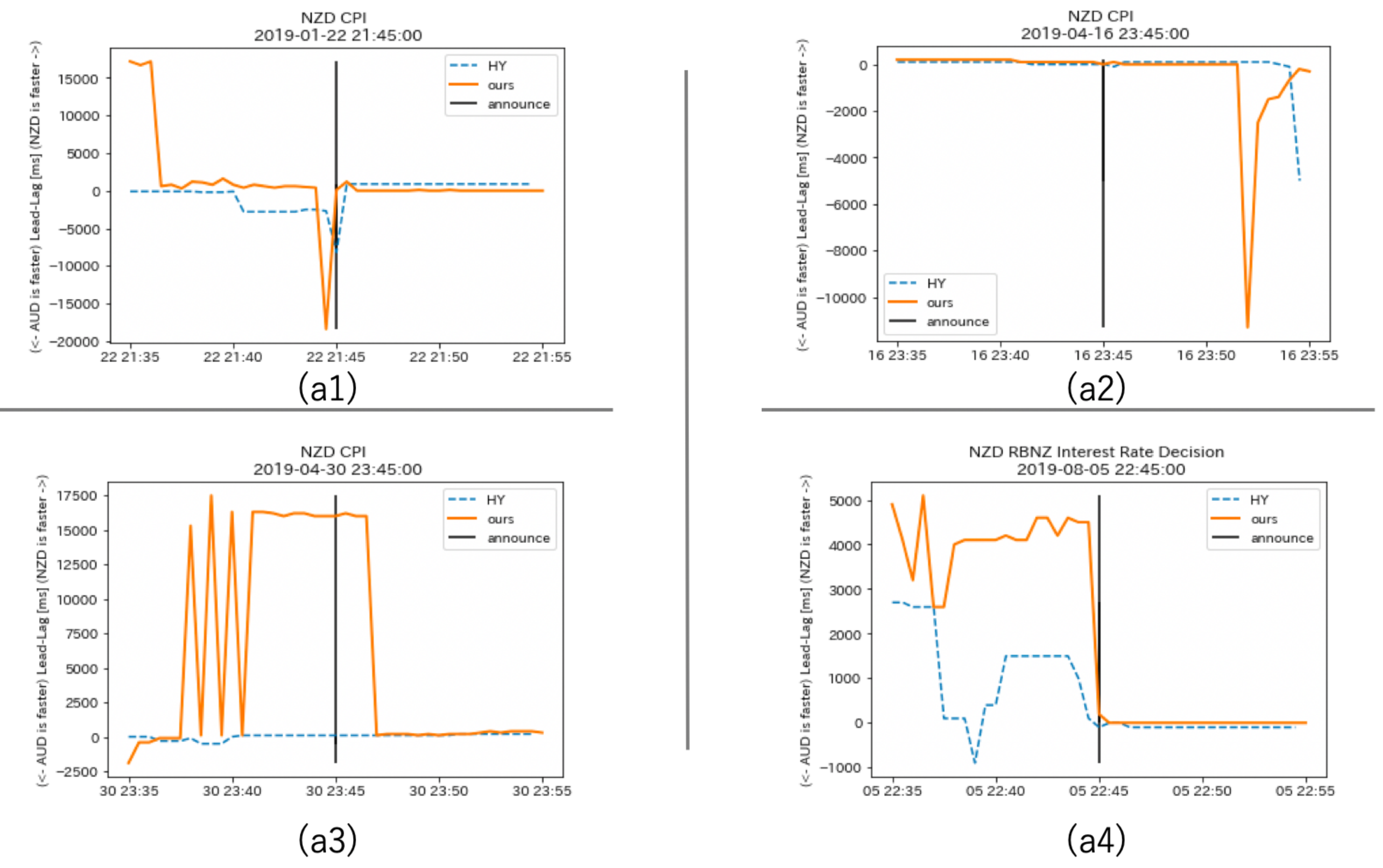}
  \caption{The important announcements of New Zealand.
   When lead-lag paramters is positive, NZD/USD is leading.}
  \label{fig:nzd}
\end{figure}

\begin{figure}
\centering
  \includegraphics[width=\textwidth]{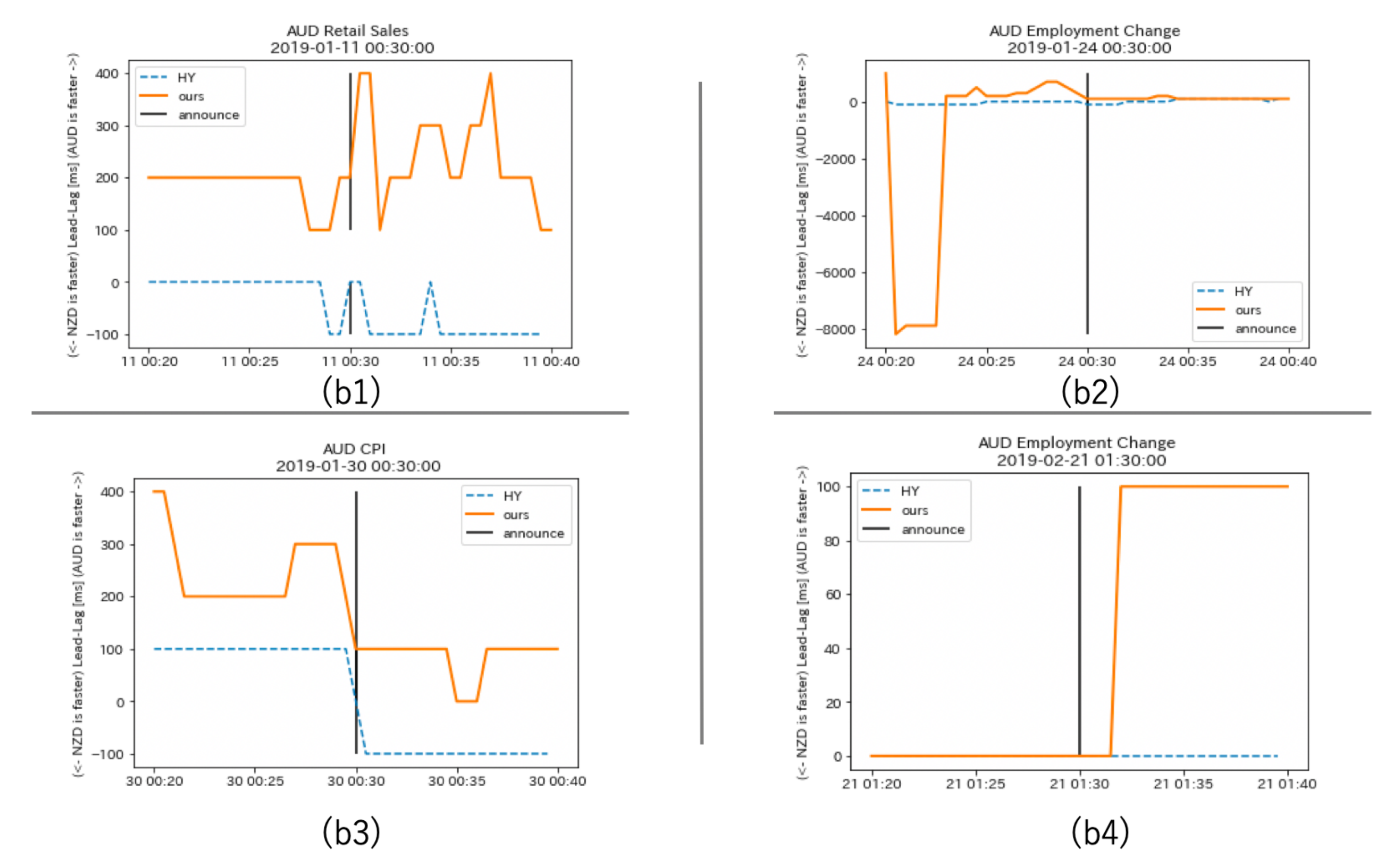}
  \caption{lead-lag parameters before and after the important announcements.
  The orange thick line shows the lead-lag estimated by our method.
  The blue dotted line shows the lead-lag estimated by Hayashi-Yoshida estimator.
 }
  \label{fig:aud}
\end{figure}

\section{Conclusion}
We investigated the estimation problems of the lead-lag effect in high-frequency data.
Mainly, we positioned the non-synchronous observation, computational complexity, and time variation of the lead-lag as critical issues.
We proposed NAPLES to solve these issues.
Our method’s estimation was with low computational complexity.
The experiment using simulated data showed that our method captures the lead-lag more accurately than the existing methods do.
The experiment using high-frequency foreign exchange market data showed that there was a lead-lag effect at the time of important announcements.


\appendix
\section{Appendix:Proof of main theorem}
The general form of the main theorem is as follows.
\begin{theorem}\label{thm:1}
Assume that two stochastic processes $X,Y$ follow
geometric Brownian motions.
\[
\left\{ \begin{array}{l} 
dX_{ t } = \sigma_1 X_t dW_t^{(1)}, \ \ X_0=x_0, \\
dY_{ t } = \sigma_2 Y_t dW_t^{(2)}, \ \ Y_0=y_0.\\
\end{array} \right.
\]
where $W _ { t-\theta } ^ { ( 2 ) }$ is the sum of  $W _ { t } ^ {(1)}$ and the independent Brownian motion $W_ { t } ^ { ( 3 ) }$ as follows
\[
W _ { t-\theta } ^ { ( 2 ) } = 
\rho W _ { t } ^ { ( 1 ) }
+
\sqrt{1-\rho^2} W_ { t } ^ { ( 3 ) }
\].
Assume that $X_t$ is observed at $0=s_1<s_2<\cdots<s_n=T$ and $Y_t$ is observed at $0=t_1<t_2<\cdots<t_m=T$.
Let $\overline{t_i}:=\min\{t_j \ | \ s_i \le t_j \}$,
$\underline{t_i}:=\max\{t_j\ | \ s_i > t_j\}$,
$\overline{s_j}:=\min\{s_i \ | \ t_j \le s_i \}$,
$\underline{s_j}:=\max\{s_i \ | \ s_i < t_j\}$ and $T=\max\{s_n,t_m\}$ then we assume that
\begin{itemize}
    \item for all $i$ there exists at most one $j$, such that $s_i < t_j < s_{i+1}$ holds,
    \item for all $j$ there exists at most one $i$ such that $t_j < s_i < t_{j+1}$ holds.
\end{itemize}
Then the  expectation of $R(T;X_t,Y_t)$ will be as follows.
\begin{eqnarray*}
\mathrm{E}[R(T;X_t,Y_t)]&=&
\frac{1}{\pi} \sum_{i=1}^{n-1} 
1_{ s_{i-1} <  \overline{t_{i}}+\theta < s_{i}}
\arcsin (
    \frac{\rho (s_i - \overline{t_{i}} - \theta)}{ \sqrt{(\underline{t_{i+1}}- \overline{t_{i}})(s_i-s_{i-1})} }
)\\
&+&
\frac{1}{\pi}
\sum_{j=1}^{m-1}
1_{\overline{s_{j}} < t_{j}+\theta < \underline{s_{j+1}}}
\arcsin (\frac{\rho (t_j + \theta - \overline{s_{j}})
}{ \sqrt{(\underline{s_{j+1}}-\overline{s_{j}})(t_{j}-t_{j-1})}  })
\\
&-&
\frac{1}{\pi} \sum_{i=1}^{n-1} 
1_{s_{i-1} < \underline{t_{i+1}}+\theta <s_{i}}
\arcsin(
    \frac{\rho  (\underline{t_{i+1}}+\theta- s_{i-1})}{ \sqrt{(\underline{t_{i+1}}- \overline{t_{i}})(s_i-s_{i-1})}
    }
    )
\\
&-&
\frac{1}{\pi}
\sum_{j=1}^{m-1}
1_{\overline{s_{j}} < t_{j-1}+\theta < \underline{s_{j+1}}}
 \arcsin(\frac{\rho (\underline{s_{j+1}}- t_{j-1}-\theta)}{\sqrt{(\underline{s_{j+1}}-\overline{s_{j}})(t_{j}-t_{j-1})} }) .
\end{eqnarray*}
\end{theorem}
First, we derive the main theorem from this theorem.
\begin{proof}[Proof of Main Theorem]
The observation times $s_i$ and $t_j$ are always equally spaced $\Delta $ ,so
\begin{enumerate}
    \item if $-\Delta<\theta<0$ then $ s_{i-1} <  \overline{t_{i}}+\theta < s_{i}$ holds for all $i$.
    \item if $0<\theta<\Delta$ then $\overline{s_{j}} < t_{j}+\theta < \underline{s_{j+1}}$ holds for all $j$.
    \item if $-2\Delta < \theta< \Delta$ then $s_{i-1} < \underline{t_{i+1}}+\theta<s_{i}$ holds for all $i$
    \item if $\Delta < \theta< 2\Delta$ then $\overline{s_{j}} < t_{j-1}+\theta<\underline{s_{j+1}}$ holds for all $j$.
\end{enumerate}
Therefore, the indicator function after sigma is 0 and we can derive the main theorem from that.\end{proof}
Second, we prepare a classic Lemma for the normal distribution needed to prove the theorem \cite{Bryc1995}.
\begin{lemma}\label{lemma:1}
For standardized bivariate normal $N, M$ with a correlation coefficient $\rho$, following arcsin law holds.
\[
\mathrm{P}(N>0,M>0)
=\frac { 1 } { 4 } + \frac { \arcsin ( \rho ) } { 2 \pi } 
= \frac { \arccos ( - \rho ) } { 2 \pi }.
\]

\end{lemma}

\begin{proof}[Proof of Theorem]
First, consider the first term.
We can assume $b^{(X)}_{s_{i+1}}=-b^{(X)}_{s_{i}}$ and $b^{(X)}_{s_{i}}\neq 0$ holds for all $ i $ because other terms do not affect $R(t)$.
By defining,
$\hat{b}^{(Y)}_{s_{i+1}}:=(\hat{Y}_{s_{i+1}}-\hat{Y}_{s_{i}} ),
\hat{b}^{(X)}_{t_{j+1}}:=(\hat{X}_{t_{j+1}}-\hat{X}_{t_{j}} )
$ we can transform the first term
$
\sum_{i=1}^{n-1} 
b^{(X)}_{s_{i}}(
\hat{Y}_{s_{i+1}}-\hat{Y}_{s_{i}})
=
\sum_{i=1}^{n-1} 
b^{(X)}_{s_{i}}\hat{b}^{(Y)}_{s_{i+1}}
$.
By the assumptions,
$
\hat{b}^{(Y)}_{s_{i+1}}= \sign(Y_{\underline{t_{i+1}}}-Y_{\overline{t_i}}),
\hat{b}^{(X)}_{t_{j+1}}= \sign(Y_{\underline{s_{j+1}}}-Y_{\overline{s_j}}),
$ holds.
Therefore, when calculating their expected value, we calculate the probabilities for the same sign and different sign respectively as follows.
\begin{eqnarray*}
\mathrm{E}[
\sum_{i=1}^{n-1} 
b^{(X)}_{s_{i}}\hat{b}^{(Y)}_{s_{i+1}}
]
&=&
\sum_{i=1}^{n-1} 
\mathrm{E}[
b^{(X)}_{s_{i}}\hat{b}^{(Y)}_{s_{i+1}}
]\\
&=&
\sum_{i=1}^{n-1} 
2\mathrm{P}(
(W^{(1)}_{s_i}-W^{(1)}_{s_{i-1}})
(W^{(2)}_{\underline{t_{i+1}}}-W^{(2)}_{\overline{t_{i}}})
>0
)+1
\end{eqnarray*}
To calculate the final term, we define the random variable as follows,
\[
p_{s_i,s_{i-1}}:=(W^{(1)}_{s_i}-W^{(1)}_{s_{i-1}}),q_{\underline{t_{i+1}},\overline{t_{i}}}:=(W^{(2)}_{\underline{t_{i+1}}}-W^{(2)}_{\overline{t_{i}}})
\]
Because $W^{(1)},W^{(2)}$ are Brownian motions,
$p_{s_i,s_{i-1}}q_{\underline{t_{i+1}},\overline{t_{i}}}$ follows a Gaussian distribution.
To use \ref{lemma:1}, we must calculate the correlation between $p_{s_i,s_{i-1}}$ and $q_{\underline{t_{i+1}},\overline{t_{i}}}$.
Here, $W^{(1)},W^{(2)}$ are Brownian motions with a lead-lag, such that the following equation holds.
\begin{eqnarray*}
q_{\underline{t_{i+1}},\overline{t_{i}}}&=&(W^{(2)}_{\underline{t_{i+1}}}-W^{(2)}_{\overline{t_{i}}})\\
&=&\rho (W_{\underline{t_{i+1}}+\theta } ^ { (1) }- W _ {\overline{t_{i}}+\theta }^{(1)})
+\sqrt{1-\rho^2} (W_{\underline{t_{i+1}}+\theta } ^ { (3) }- W _ {\overline{t_{i}}+\theta }^{(3)})\\
&=&\rho p_{\underline{t_{i+1}}+\theta,\overline{t_{i}}+\theta} + \sqrt{1-\rho^2}r
\end{eqnarray*}
where $r= (W_{\underline{t_{i+1}}+\theta } ^ { (3) }- W _ {\overline{t_{i}}+\theta }^{(3)})$ is a random variable, which is independent of $p$.
Therefore, we can calculate the covariance by the bi-linearity of the covariance and the independence of $ r$.
\begin{eqnarray*}
\mathrm{Cov}(p_{s_i,s_{i-1}},q_{\underline{t_{i+1}},\overline{t_{i}}})
&=&\mathrm{Cov}(p_{s_i,s_{i-1}},
\rho  p_{\underline{t_{i+1}}+\theta,\overline{t_{i}}+\theta} + \sqrt{1-\rho^2}r)\\
&=& 
\rho \mathrm{Cov}(p_{s_i,s_{i-1}},
p_{\underline{t_{i+1}}+\theta,\overline{t_{i}}+\theta}
) + \mathrm{Cov}(p_{s_i,s_{i-1}},\sqrt{1-\rho^2}r)\\
&=& \rho \mathrm{Cov}(p_{s_i,s_{i-1}},
p_{\underline{t_{i+1}}+\theta,\overline{t_{i}}+\theta}
) = \rho 
\mathrm{Cov} (
(W^{(1)}_{s_i}-W^{(1)}_{s_{i-1}}),
(W_{ \underline{t_{i+1}}+\theta } ^ { (1) }- W _ { \overline{t_{i}}+\theta}^{(1)})
)
\end{eqnarray*}
From the independent increment property of Brownian motion,
\[
\mathrm{Corr}(p_{s_i,s_{i-1}},
q_{\underline{t_{i+1}},\overline{t_{i}}})=
\begin{cases}
    \rho (s_i - \overline{t_{i}} - \theta) / \sqrt{(\underline{t_{i+1}}- \overline{t_{i}})(s_i-s_{i-1})} 
    & \text{for } s_{i-1} <  \overline{t_{i}}+\theta < s_{i}  \\
    \rho  (\underline{t_{i+1}}+\theta- s_{i-1}) / \sqrt{(\underline{t_{i+1}}- \overline{t_{i}})(s_i-s_{i-1})}  
    & \text{for } s_{i-1} < \underline{t_{i+1}}+\theta < s_{i}\\
    0 & \text{otherwise}
\end{cases}
\]
is obtained.
By using the lemma, we obtain the following equation.
\[
\mathrm{P}(p_{s_i,s_{i-1}}q_{\underline{t_{i+1}},\overline{t_{i}}}>0)
=
\begin{cases}
     \frac{1}{\pi} \arcsin (
    \frac{\rho (s_i - \overline{t_{i}} - \theta)}{ \sqrt{(\underline{t_{i+1}}- \overline{t_{i}})(s_i-s_{i-1})} }
     ) + \frac{1}{2} 
    & \text{for } s_{i-1} <  \overline{t_{i}}+\theta < s_{i} \\
     \frac{1}{\pi}  \arcsin(
    \frac{\rho  (\underline{t_{i+1}}+\theta- s_{i-1})}{ \sqrt{(\underline{t_{i+1}}- \overline{t_{i}})(s_i-s_{i-1})}
    }
    ) + \frac{1}{2} 
    & \text{for } s_{i-1} < \underline{t_{i+1}}+\theta < s_{i}\\
    \frac{1}{2}  & \text{otherwise}
\end{cases}
\]
This is repeated for the second term and the theorem is obtained.
\end{proof}

\end{document}